\documentclass[10pt]{article}

\usepackage[utf8]{inputenc}
\usepackage[utf8]{inputenc}
\usepackage[english]{babel}
\usepackage{amsmath}
\usepackage{amssymb}
\usepackage{amsthm}
\usepackage[T1]{fontenc}
\usepackage[utf8]{inputenc} 
\usepackage{tikz}
\usepackage{verbatim}
\usetikzlibrary{arrows}
\usepackage{float}
\usepackage{tkz-berge}
\usepackage{subcaption}
\usepackage{caption}
\usepackage{array}
\usepackage{bm}
\usepackage[labelfont=bf]{caption}
\usepackage{enumerate}
\usepackage[final]{pdfpages}
\usepackage{graphicx}
\usepackage{epigraph}
\usepackage{xtab}
\usepackage{color}
\usepackage{framed}
\usepackage{lipsum} 
\usepackage{geometry} 
\usepackage{marginnote} 

\definecolor{myRed}{rgb}{1,0,0}
\definecolor{myBlue}{rgb}{0,0,1}
\definecolor{myYellow}{rgb}{0,1,0}
\tikzstyle{red}=[thick, myRed, opacity=1]
\tikzstyle{blue}=[thick, myBlue, opacity=1]
\tikzstyle{yellow}=[ thick, myYellow, opacity=1]
\tikzset{edge/.style = {->,> = latex'}}

\newcommand{\path}{\vec{p}}

\newcolumntype{L}{>{\arraybackslash}m{3cm}}

\newcommand{\normal}{\mathcal{N}}

\renewcommand{\P}{\mathbb{P}}
\newcommand{\E}{\mathbb{E}}
\newcommand{\V}{\mathbb{V}}

\newcommand{\N}{\mathbb{N}}
\newcommand{\R}{\mathbb{R}}

\newcommand{\1}{\boldsymbol{1}}
\renewcommand{\ij}{(i,j)}
\newcommand{\rP}{\vec{\mathcal{P}}}

\def\var#1{\mathbb{V}\left[#1\right]}

\newtheorem{thm}{Theorem}

\newtheorem{ex}[thm]{Example}

\newtheorem{ass}[thm]{Assumption}
\newtheorem{remark}{Remark}

\title{Quality analysis in acyclic production networks} 

\author
{Abraham Gutierrez, Sebastian M\"uller
\\
\today
}

\date{}

\begin{document}


\maketitle

%
\begin{abstract}
The production network under examination consists of a number of workstations. Each workstation is a parallel configuration of machines performing the same kind of tasks on a given part. Parts move from one workstation to another and at each workstation a part is assigned randomly to a machine. We assume that the production network is \emph{acyclic}, that is, a part does not return to a workstation where it previously received service. Furthermore,  we assume that the quality of the end product is \emph{additive}, that is, the sum of the quality contributions of the machines along the production path. The contribution of each machine  is modeled by a separate random variable. 

Our main result is the construction of estimators that allow pairwise and multiple comparison of the means and variances of machines in the same workstation. These comparisons then may lead to the identification of unreliable machines.  We also discuss the asymptotic distributions of the estimators that allow the use of standard statistical tests and decision making.



\end{abstract}


\section{Introduction}

In order to maintain  competitiveness, industrial manufacturers have to pursue new sources to enhance process agility. In parallel production networks, the quality of the end product depends on various intermediate production steps. Therefore,  a recent  challenge is to achieve a level of visibility into the production flows that allow to optimize throughput by guaranteeing at the same time given quality standards. The challenge to maintain a visibility across all parallel workflows and  to identify eventual sources of errors becomes particularly difficult in situations where the qualities of single machines are not observable.  

\begin{figure}
\centering
\tikzset{
node_style/.style={circle,draw=black,fill=black!20!, minimum size=0pt},
node_stylethick/.style={circle,draw=black,fill=black, minimum size=0pt},
edge_style/.style={draw=black,  thin, ->},
edge_stylethick/.style={draw=black,  ultra thick, ->},
label/.style={sloped,above}
}
\begin{tikzpicture}[shorten >=1pt,
   auto,
   node distance=2cm and 10cm,
   scale=0.6]

   \draw [rounded corners=5, fill=black!10!] (2.5,-7) rectangle ++(1,8) node [midway] {};
   \node[text width=1cm] at (3.15,1.3) {WS1};
    \node[text width=1cm] at (6.15,1.3) {WS2};
     \node[text width=1cm] at (9.15,1.3) {WS3};
      \node[text width=1cm] at (12.15,1.3) {WS4};
   
   \draw [rounded corners=5, fill=black!10!] (5.5,-7) rectangle ++(1,8) node [midway] {};

   \draw [rounded corners=5, fill=black!10!] (8.5,-7) rectangle ++(1,8) node [midway] {};

   \draw [rounded corners=5, fill=black!10!] (11.5,-7) rectangle ++(1,8) node [midway] {};
   
    \node[node_stylethick] (s) at (0,-3) {};
     \foreach \y/\t in {1/1, 2/1, 3/1, 4/1}{
     \node[node_style]  (v\y) at (3, -2*\y+2) {};
     \draw [edge_style]  (s) edge  (v\y);}
       
       \draw [edge_stylethick]  (s) edge  (v1);
 \node[node_stylethick] (v1) at (3,0) {};
       
     \foreach \y/\t in {1/2, 2/2, 3/2}
     \node[node_style]  (w\y) at (6, -2*\y+1) {};
     \foreach  \y in {1,2,3,4} 
        \foreach  \w in {1,2,3} 
           \draw [edge_style]  (v\y) edge  (w\w);
           
            \draw [edge_stylethick]  (v1) edge  (w3);
            \node[node_stylethick] (w3) at (6,-5) {};
           
     \foreach \y/\t in {1/3, 2/3}
     \node[node_style]  (x\y) at (9, -2*\y) {};
     \foreach  \y in {1,2,3} 
        \foreach  \w in {1,2} 
           \draw [edge_style]  (w\y) edge  (x\w);
           
            \draw [edge_stylethick]  (w3) edge  (x1);
           \node[node_stylethick] (x1) at (9,-2) {};
           
      \foreach \y/\t in {1/4, 2/4, 3/4, 4/4}
     \node[node_style]  (y\y) at (12, -2*\y+2) {};
     \foreach  \y in {1,2} 
        \foreach  \w in {1,2,3,4} 
           \draw [edge_style]  (x\y) edge  (y\w);
           
            \draw [edge_stylethick]  (x1) edge  (y4);
            \node[node_stylethick] (y4) at (12,-6) {};
      
      \node[node_stylethick] (t) at (15,-3) {$$};     
           \foreach  \w in {1,2,3,4} 
            \draw [edge_style]   (y\w) edge (t);
            
            \draw [edge_stylethick]  (y4) edge  (t);
    \end{tikzpicture}
\caption{An illustration of a path passing through different workstations in a production network.}\label{fig:path}
\end{figure}
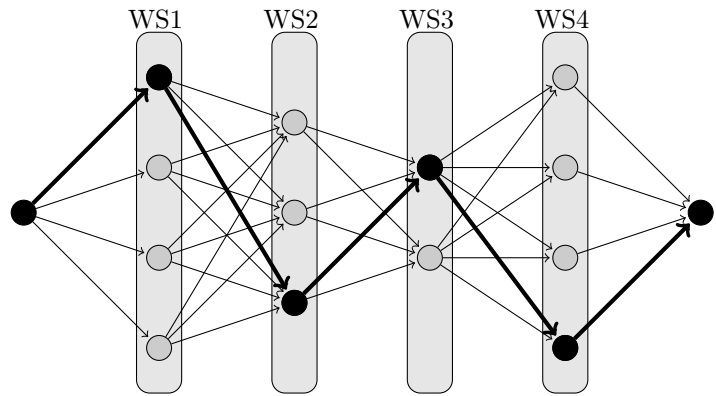

In this note we discuss a new approach allowing to compare the impact of different machines of the same workstation, i.e.~performing the same task in parallel, on the quality of the  product. We assume that the quality of the product is only observable at the end of the process and that the  qualities of the nodes along the path are latent variables. See Figure \ref{fig:path} for an illustration.
Although the main interest is usually in the process mean, the size of the process variability is often crucial.  Our approach allows to model the process mean and process variability in parallel under very general conditions.  In particular, the only condition on the production network is that it is acyclic, i.e.~a part does not return to a workstation where it previously received service. These kind of networks are modeled in mathematics and computer science as  directed acyclic graphs (DAGs). 

The concept of DAGs is becoming increasingly important having various applications in different fields of science and engineering. For instance DAGs have many applications not only in  modeling  production networks, e.g.~\cite{HaLe:89},  but also in process networks and  distributed computing, e.g.~\cite{FeRuSch:97},	scheduling for systems of tasks, e.g.~\cite{Sk:08},  DAG-networks for deep learning, e.g.~\cite{Sz:15}, and DAG-based alternatives of the blockchain technology, e.g.~\cite{Po:18}. Moreover, directed acyclic graphs (DAGs) are more and more often used to represent  causal relationships among random variables in  complex systems, e.g.~\cite{Pe:09}. We speak in general about production networks, but due to the wide applications of DAGs we want to emphasize here that our approach can be applied to any kind of acyclic network where effects along a path are additive and where measurements can only be performed at the end of the process. 

Parameter estimation is not only important for the decision making process, but it  is also an essential prerequisite in order to obtain meaningful simulations. However, complex production networks depend on a high number of parameters and their estimation is  challenging. Since simulations in manufacturing plays an increasingly important role, e.g.~\cite{ KoPh:01, MoDoBe:14}, we believe that our approach is an interesting contribution to this topic as well.

\subsection*{Previous work}
Estimation for the mean can also be conducted using a linear regression with categorical covariates, e.g.~\cite{MoPeVi:01}. 
However, since in a linear regression homoscedasticity, i.e.~the homogeneity of the variances,  is a crucial condition, this approach naturally does not allow to compare differences of the variances.
In fact, linear regression models are often plagued by different variabilities or heteroskedasticity. We refer to \cite[Chapter 4]{KlZe:08} for an overview on how to detect and control heteroskedasticity. In contrast to these problems, our approach naturally allows different  variabilities  between variables  and is able to detect differences in variability of values of a given variable.  

 Multifactor experimental designs,  \cite{DrPu:96, Se:18}  are also alternatives to estimate the mean differences but also rely on homoscedasticity.  They are mostly used in the context of statistically planned experiments, which consists of a few experimental runs to obtain data on the product characteristics. If the number of observations for each setting (or path in our notation) is sufficiently high and under further conditions described in \cite[Section 4]{DrPu:96} these methods allow a comparison of the variances, too. While this may offer a feasible, however not direct, way to identify differences in variability if the number of machines is small, it seems not practical in more complex networks.

There is also a connection to critical paths analysis, see e.g.~\cite{BoGeWeAr:10, Schu:05}. While these methods allow to find critical paths in acyclic networks they are not suited to compare nor estimate differences in mean and variances of given tasks.

\section{The model}
A directed acyclic graph (\textbf{DAG}) is a finite directed graph with no directed cycles. It consists of a finite vertex set $V$ and a finite set of directed edges $E=\{ (v,w): v,w \in V, v\neq w\}$. In our setting the DAG contains two special vertices: a source $s$ and a sink $t$.  We are interested in the paths from the source to the sink in this graph.  We denote a path $\vec{p}$ in the DAG as $\vec{p} = (p_{0},p_{1},\ldots, p_{c}, p_{c+1})$ where $p_{0}=s$ and $p_{c+1}=t$ and $(p_{i},p_{i+1})\in E$. We define $\vec{p}[j]:=p_{j}$.  We refer to Figure \ref{fig:DAG1} for an illustration and to \cite{BaGu:09} for more details on directed graphs.

We assume that at each step $1\leq i \leq c$ the path $\vec{p}$ has  $r_{i}\leq r$  different choices and the nodes in each column are always numerated starting with $1$.  The possible choices of a path can therefore be modeled through an $r\times c$ matrix.   More precisely, given a path $\vec{p}$, we associate an $r \times c$ binary matrix $V_{\vec{p}}$ that has $1$'s 
only in the nodes visited by the path:
$$
V_{\vec{p}} := (V_{\vec{p}}(i,j))_{i\in[r],\, j\in [c]},
$$
where we denote $[k]:=\{1,2,\ldots, k\}$ for an integer $k$.
We call $V_{\vec{p}}$ the \textbf{indicator matrix} of the path $\vec{p}$.

Each path contains exactly one node of each column.  The aim of this paper is to study differences among nodes of the same columns.
We think of nodes in the same columns as different possibilities for a given task,  as different persons performing the same job,  as different machines in the same workstation or  as  variations of the same kind of treatment.  

\begin{figure}
\centering
\tikzset{
node_style/.style={circle,draw=black,fill=black!20!, minimum size=1cm},
edge_style/.style={draw=black,  thick, ->},
label/.style={sloped,above}
}
\begin{tikzpicture}[shorten >=1pt,
   auto,
   node distance=2cm and 10cm,
   scale=0.9]
    \node[node_style] (s) at (0,-3) {$s$};
     \foreach \y/\t in {1/1, 2/1, 3/1, 4/1}{
     \node[node_style]  (v\y) at (3, -2*\y+2) {(\y,\t)};
     \draw [edge_style]  (s) edge  (v\y);}
       
     \foreach \y/\t in {1/2, 2/2, 3/2}
     \node[node_style]  (w\y) at (6, -2*\y+1) {(\y,\t)};
     \foreach  \y in {1,2,3,4} 
        \foreach  \w in {1,2,3} 
           \draw [edge_style]  (v\y) edge  (w\w);
           
     \foreach \y/\t in {1/3, 2/3}
     \node[node_style]  (x\y) at (9, -2*\y) {(\y,\t)};
     \foreach  \y in {1,2,3} 
        \foreach  \w in {1,2} 
           \draw [edge_style]  (w\y) edge  (x\w);
           
      \foreach \y/\t in {1/4, 2/4, 3/4, 4/4}
     \node[node_style]  (y\y) at (12, -2*\y+2) {(\y,\t)};
     \foreach  \y in {1,2} 
        \foreach  \w in {1,2,3,4} 
           \draw [edge_style]  (x\y) edge  (y\w);
      
      \node[node_style] (t) at (15,-3) {$t$};     
           \foreach  \w in {1,2,3,4} 
            \draw [edge_style]   (y\w) edge (t);
    \end{tikzpicture}
\caption{An illustration of a DAG with $c=4$ and  $r_{1}=4, r_{2}=3, r_{3}=2$ and $r_{4}=4$. Every node in column $i$ has outgoing edges to every node in column $i+1,~ i=1,\ldots, c-1$.}\label{fig:DAG1}
\end{figure}
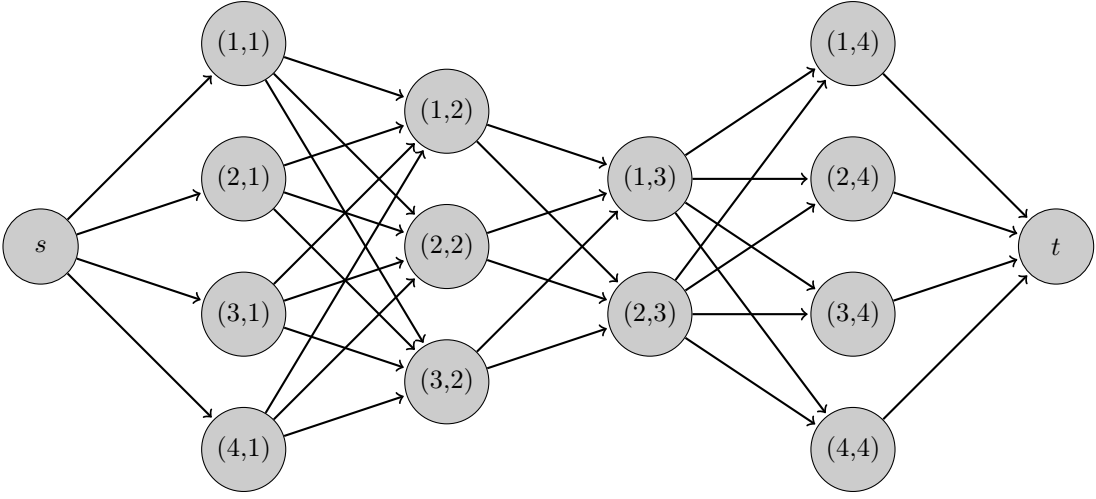

The given data consists of the list of 
paths $\{ \vec{p}_{i} \}_{i = 1,\ldots, n}$ in the DAG and the list of outputs $\{ b(\vec{p}_{i}) \}_{i = 1,\ldots, n}$ for each path.

We consider the {\bf quality matrix} $S$, which is a random matrix of size $r\times c$ with real entries
$$
S := ( s(i,j) )_{i\in[r],\, j\in [c]},
\qquad
s(i,j) \in \mathbb{R}.
$$
We model the paths with a random vector $\rP := (P_{1}, \ldots, P_{c})$ 
where the components $P_{i}$ are random variables over the set $[r]$. 

Throughout the paper we work under the following standing assumptions:
\begin{ass} We assume that:
\begin{enumerate}
 \item   all entries of $S$ have finite second moments; 
 \item  the random variables $S\ij, i\in[r],\, j\in [c]$ are (jointly) independent;
 \item the paths $\rP_{1}, \rP_{2}, \ldots$ are chosen independently and uniformly.
\end{enumerate}
\end{ass}
\noindent Note that,  we do not assume the entries of $S$ to be identically distributed nor having the same variance.

Let $\vec{p} = (p_{0},p_{1},\ldots, p_{c}, p_{c+1})$ be a realization of $\rP$, where $p_{0}=s$ and $p_{c+1}=t$ almost surely.
Then, the {\bf quality of the construction path} $\vec{p}$ is defined as
$$
b(\vec{p})= \sum_{j= 1}^{c} S(p_{j}, j).
$$
We also can think of $b(\vec{p})$ as the quality (or error) cumulated along the path $\vec{p}$.

Let us make precise where the randomness enters in our model. We choose a random path $\rP$ and random matrix $S$. The corresponding probability measure is denoted by $\P$.
 The random choice of $\rP$ and $S$ induces a random variable $b(\rP)$ and allows to generate a sequence of  i.i.d.~random variables $(\rP_{1}, b(\rP_{1})), (\rP_{2}, b(\rP_{2})),\ldots$.

 The goal of our study is to give estimates on the law of $S$ by observing the paths $\rP$ and its cumulated qualities $b(\rP)$. Note that,   $(\rP_{n}, b(\rP_{n}))_{n\in\N}$ is  in general not a sufficient statistic for $S$, i.e.~we can not recover the distribution of $S$ by only observing realizations of $(\rP, b(\rP))$, as we see in the next remark

\begin{remark}\label{rem:not_sufficient}
Let us consider the case $r=1$ and $c=2$. Let $S(1,1) \sim\normal(0,1)$ and  
 $S(1,2) \sim\normal(1,1)$ and define  $\tilde S(1,1):= S(1,1)+1$ and  
 $\tilde S(1,2):= S(1,2)-1$. Then for any given path $\path$ we have that
$\sum_{j= 1}^{2} S(p_{j}, j)=  \sum_{j= 1}^{2} \tilde S(p_{j}, j)$. Hence, the statistic $(\path_{n}, b(\path_{n}))_{n\in \N}$ does not allow us to distinguish between $S$ and $\tilde S$. 
\end{remark}

\begin{ex}[Binary errors]
The matrix $S$ consists of  independent Bernoulli random variables $S(i,j)$. The value $1$ of this Bernoulli may encode a defect and hence $b(\path)$ counts the number of defects of the end product.
\end{ex}

\begin{ex}[Gaussian quality]\label{ex:gaussian}
The matrix $S$ consists of  independent Gaussian  random variables $S(i,j)$. The end quality $b(\path)$ is then distributed as a mixture of Gaussian random variables.
 \end{ex}

Given a sequence of realizations $(\vec{p}_{k})_{k\in[n]}$ of $\rP$, we define the following matrices that at are the core of our analysis:
\begin{equation*}
B^{(n)} := \sum_{k=1}^{n} b(\vec{p}_{k}) V_{\vec{p}_{k}},~
V^{(n)} := \sum_{k=1}^{n} V_{\vec{p}_{k}}, n\geq 1.
\end{equation*}
The value $B^{(n)}(i,j)$ is the sum of all cumulated qualities of paths containing node $\ij$, whereas $V^{(n)}(i,j)$ just counts the number of times node $\ij$ was used.
We define the \textbf{sample mean matrix} as the sample mean quality matrix  :
$$
T^{(n)} := (T^{(n)}(i,j))_{i\in[r],\, j\in [c]}, 
\mbox{ where }
T^{(n)}(i,j)
:=
\begin{cases}
\frac{B^{(n)}(i,j)}{V^{(n)}(i,j)},&\quad\text{if } V^{(n)}(i,j) \neq 0 ;\\
\text{0, } &\quad\text{otherwise}.\\
\end{cases}
$$

The corresponding {\bf sample variance matrix} $\Sigma^{(n)}$ is defined by
\begin{equation*}
\Sigma^{(n)}(i,j):= \begin{cases}
\frac{1}{V^{(n)}(i,j)} \sum_{k=1}^{n} \left(b(\vec{p}_{k}) V_{\vec{p}_{k}}(i,j) - T^{(n)}(i,j)\right)^{2},&\quad\text{if } V^{(n)}(i,j) \neq 0 ;\\
\text{0, } &\quad\text{otherwise}.\\
\end{cases}
\end{equation*}

\section{Results}\label{section:Results}
Denote $D^{(n)}=(\rP_{i})_{i\in[n]}$ the multi-set\footnote{We use a multi-set since we  need to keep track of the multiplicity of the  paths.} or sequence of all paths up to time $n$ and let $$D^{(n)}_{(i,j)} := \{ \path \in D^{(n)} \, : \, \path[j] = i \}$$
	be the  multi-set of all paths up to index $n$ that go through the node $(i,j)$; we note that we can recover $V^{(n)}(i,j)$ through $D^{(n)}_{(i,j)}$ by  $V^{(n)}(i,j) = |D^{(n)}_{(i,j)}|$.
	
In general it is not possible to estimate the mean quality matrix $\E[S]$, see Remark \ref{rem:not_sufficient}.  However,  it is possible to identify nodes with higher or lower quality mean or variance in each column.
	
\begin{thm}\label{thm:lln}
Let $(i,j), (i', j) \in [r] \times [c]$, then 
$$
	T^{(n)}(i,j) - T^{(n)}(i',j)\xrightarrow[n\rightarrow \infty]{a.s.} \E[S(i,j)] - \E[S(i',j)]
$$
and 
$$
	\Sigma^{(n)}(i,j) - \Sigma^{(n)}(i',j)\xrightarrow[n\rightarrow \infty]{a.s.} \V[S(i,j)] - \V[S(i',j)].
$$
\end{thm}
\begin{proof} Using twice the law of large numbers and the continuous mapping theorem we obtain
\begin{eqnarray*}
T^{(n)}(i,j) 
&= &
\frac{1}{|D^{(n)}_{(i,j)}|} 
\sum_{\vec{p} \in D^{(n)}_{(i,j)} }
b(\vec{p}) \cr
&=& \frac{n}{|D^{(n)}_{(i,j)}|} \frac1n
\sum_{\vec{p} \in D^{(n)}_{(i,j)} }
b(\vec{p}) \cr
&\xrightarrow[n\rightarrow \infty]{a.s.}& \frac{1}{\P(\rP[j]=i)} \E[b(\rP); \rP[j]=i'] = \E\left[b(\rP) \, | \, \rP[j]=i\right].
\end{eqnarray*}
In the same way
\begin{equation*}
T^{(n)}(i',j) 
= 
\frac{1}{|D^{(n)}_{(i',j)}|} 
\sum_{\vec{p} \in D^{(n)}_{(i',j)} }
b(\vec{p})
\xrightarrow[n\rightarrow \infty]{a.s.}
\E\left[b(\rP) \, | \, \rP[j]=i'\right].
\end{equation*}
Using the assumption that the paths are chosen uniformly and the definition of $b(\rP)$, we obtain the first part of the theorem from
\begin{eqnarray*}
\E\left[b(\rP)| \, \rP[j]=i\right] - \E\left[b(\rP) \, | \, \rP[j]=i'\right] &=& \frac{\E\left[b(\rP);  \rP[j]=i\right]- \E\left[b(\rP); \rP[j]=i'\right] }{\P(\rP[j]=i)} \cr
&=& \frac{\E\left[S(i,j);  \rP[j]=i\right]- \E\left[S(i',j); \rP[j]=i\right] }{\P(\rP[j]=i)} \cr
&=& \E[S(i,j)] - \E[S(i',j)].
\end{eqnarray*}
For the second part of the theorem, we use the law of large numbers and the continuous mapping theorem to get that
\begin{eqnarray*}
\Sigma^{(n)}(i,j) 
&=&  
\left(\frac{1}{|D_{ij}^{(n)}|} \sum_{\vec{p} \in D_{ij}^{(n)} }  b^2(\vec{p})
\right)
-  \left(
T^{(n)}(i,j)
\right)^{2}
\cr
&\xrightarrow[n\rightarrow \infty]{a.s.}&
\E\left[b(\rP)^{2}| \, \rP[j]=i\right] -\E\left[b(\rP)| \, \rP[j]=i\right]^{2}.
\end{eqnarray*}
Using the definition of  $b(\path)$ we deduce  with elementary calculations that
$$
\E\left[b(\rP)^{2}| \, \rP[j]=i\right] -\E\left[b(\rP)| \, \rP[j]=i\right]^{2}
=
A_j
+
\var{S(i,j)}
$$
where $A_j$ is a quantity that only depend on the column $j$. Applying this identity for $(i,j)$ and $(i',j)$ we obtain that
$$
	\Sigma^{(n)}(i,j) - \Sigma^{(n)}(i',j)
	\xrightarrow[n \rightarrow \infty]{a.s.}
	\var{S(i,j)} - \var{S(i',j)}.
$$
\end{proof}

\subsection{Asymptotic distribution}
The sum  $$\sum_{\vec{p} \in D^{(n)}_{(i,j)} }
b(\vec{p})=\sum_{k=1}^{n}  b(\rP_{k}) \1\{ \rP_{k}[j]=i\} $$ can be interpreted as the sum of random variables  appearing in an acception-rejection method. More precisely, we start with $k=1$ and consider $(\rP_{k}, b(\rP_{k}))$. If  $\rP_{k}[j]=i$ we set $Y:=b(\rP_{k})$ and stop, otherwise we increase $k$ and repeat until $\rP_{K}[j]=i$ for the first $K$. Now, for $y\in \R$,
\begin{equation*}
\P(  Y\leq y)= \sum_{k=1}^{\infty} \P(b(\rP_{k}) \leq y  \, | \, K=k) \P(K=k)=  \P(b(\rP_{1}) \leq y  \, | \, \rP_{1}[j]=i).
\end{equation*}
In other words, the distribution of $Y$ equals the distribution of $b(\rP_{1})$ conditioned on $\rP_{1}[j]=i$. Iterating this acception-rejection method we  see that $|D^{(n)}_{(i,j)}|$ describes the number of acceptions using $(\rP_{k}, b(\rP_{k}))$, $1\leq k \leq n$. Hence the estimator $T^{(n)}$ has the same distribution as 
\begin{equation*}
\frac{1}{|D^{(n)}_{(i,j)}|} 
\sum_{k=1}^{|D^{(n)}_{(i,j)}|} Y_{k},
\end{equation*}
where $Y_{k}, k\in\N,$ is a sequence of  i.i.d.~random variables distributed as $b(\rP_{1})$ conditioned on $\rP_{1}[j]=i$.
Finally, Anscombe's theorem, \cite[Theorem 1.3.1]{Gu:09}, implies that
\begin{equation*}
\sqrt{|D^{(n)}_{(i,j)}|} \left(T^{(n)}(i,j) -  \mu_{i,j}\right) \xrightarrow[n\rightarrow \infty]{\mathcal{D}} \normal(0,\sigma_{i,j}^{2}),
\end{equation*}
where
\begin{equation*}
\mu_{i,j}:= \E\left[b(\rP) \, | \, \rP[j]=i\right] \mbox{ and } \sigma_{i,j}^{2}:=\E\left[b(\rP)^{2}| \, \rP[j]=i\right] -\E\left[b(\rP)| \, \rP[j]=i\right]^{2}.
\end{equation*}
For the variance we assume that the entries of $S$ have finite forth moments.  The estimator for the variance is
\begin{equation*}
\Sigma^{(n)}(i,j)= \frac{1}{|D_{(i,j)}|} \sum_{\vec{p} \in D_{(i,j)}^{(n)} } \left( b(\vec{p})-  \frac{1}{|D_{(i,j)}^{(n)}|}  \sum_{\vec{p} \in D_{(i,j)}^{(n)}} b(\vec{p}) \right)^{2}.
\end{equation*}
Since the distribution of $\Sigma^{(n)}(i,j)$ does not change if we replace $b(\vec{p})$ by $b(\vec{p})-\mu_{i,j}$ we can assume that $\mu_{i,j}=0$. Moreover,
\begin{equation*}
\Sigma^{(n)}(i,j)= \frac{1}{|D_{(i,j)}|} \sum_{\vec{p} \in D_{(i,j)}^{(n)} }  b(\vec{p})^{2}-  \left(T^{(n)}(i,j)\right)^{2}.
\end{equation*}
We have that $T^{(n)}(i,j)$ converges almost surely to $0$ and $\sqrt{D_{(i,j)}^{(n)}} T^{(n)}(i,j)$ converges in distribution to $\normal(0,\sigma_{i,j}^{2})$.  Slutzky's theorem implies that $\sqrt{D_{(i,j)}^{(n)}} T^{(n)}(i,j)^{2}$ converges in distribution to $0$ and therefore also in probability to $0$.
Using  Slutzky's theorem and Anscombe's theorem as above we obtain that
\begin{eqnarray*}
\sqrt{D_{(i,j)}^{(n)}} \left( \Sigma^{(n)}(i,j)-\sigma_{i,j}^{2}\right) & =&  \sqrt{D_{(i,j)}^{(n)}}\left(  \frac{1}{|D_{(i,j)}|} \sum_{\vec{p} \in D_{(i,j)}^{(n)} }  b(\vec{p})^{2}  -\sigma_{i,j}^{2}\right) \cr
&\xrightarrow[n\rightarrow \infty]{\mathcal{D}}& \normal(0,\tau_{i,j}^{2}),
\end{eqnarray*}
 with $\tau_{i,j}^{2}=\E\left[ (b(\rP)-\mu_{i,j})^{4} \, | \, \rP[j]=i\right]- \E\left[ (b(\rP)-\mu_{i,j})^{2} \, | \, \rP[j]=i\right]^{2}.$

\subsection{Pairwise and multiple comparison}
The asymptotic distributions of the estimators $T^{(n)}$ and $\Sigma^{(n)}$  justify that pairwise comparison of the mean quality and the variance of quality can be done using standard tests, e.g.~$t$-test, Fisher tests and Bartlett tests. We sketch only  the construction of a test statistic  to compare the means of two nodes under the hypothesis that they have the same variance. However, it is standard to extent this result to the case where the variances are not equal and to the estimators of the variance differences.  More details on possible applications are given in Section \ref{Examples}.
We have for $n$ sufficiently large 
\begin{equation*}
\frac{\sqrt{|D_{(i,j)}^{(n)}|}}{\sigma_{i,j}} (T^{(n)}(i,j) - \mu_{i,j}) \sim \normal(0, 1)
\mbox{ and } \frac{\sqrt{|D_{i',j}^{(n)}|}}{\sigma_{i',j}} (T^{(n)}(i',j) - \mu_{i',j}) \sim \normal(0, 1).
\end{equation*}

\noindent Under the hypothesis that  $\mu_{i,j}=\mu_{i',j}$ and $\sigma=\sigma_{i,j}=\sigma_{i',j}$ one obtains that  for $n$ sufficiently large 
\begin{equation*}
\left(\frac{|D_{i,j}^{(n)}| \cdot |D_{i',j}^{(n)}|}{|D_{i,j}^{(n)}|+|D_{i',j}^{(n)}|} \right)^{1/2} 	\frac{T^{(n)}(i,j) - T^{(n)}(i',j)}{\sigma} \sim \normal(0,1).
\end{equation*}
Replacing $\sigma$ by its estimator 	
\begin{equation*}
\hat \sigma = \left(    \frac1{|D_{i,j}^{(n)}|+|D_{i',j}^{(n)}|} \left(|D_{i,j}^{(n)}| \Sigma^{(n)}(i,j)+ |D_{i',j}^{(n)}| \Sigma^{(n)}(i',j)\right) \right)^{1/2},
\end{equation*}
we find that
\begin{equation*}
\left(\frac{|D_{i,j}^{(n)}| \cdot |D_{i',j}^{(n)}|}{|D_{i,j}^{(n)}|+|D_{i',j}^{(n)}|} \right)^{1/2} 	\frac{T^{(n)}(i,j) - T^{(n)}(i',j)}{\hat \sigma}
\end{equation*}
can be approximated by a standard normal distribution for $n$ sufficiently large.

The identification of nodes with bad quality boils down to multiple comparisons.   For a given column $j$ with $r_{j}$ nodes we define a vector $X=(X_{1},\ldots, X_{r_{j}-1})$ of test statistics
\begin{equation*}
X_{k}= T^{(n)}(k,j)- T^{(n)}(r_{j}, j)\quad k\in\{1,\ldots, r_{j}-1\}.
\end{equation*}
The vector $X$ satisfies (asymptotically) the positive regression dependency. 
Therefore, we suggest  the Benjamini-Yekutieli method, see \cite{BeYe}, to control the false discovery rate.

\section{Examples}\label{Examples}

\subsection{Wafer production}
Our study was motivated by a root-cause analysis in the wafer fabrication. Wafer fabrication is in general a procedure of many repeated sequential processes.  For instance, a simplified illustration consists of $12$ subsequent fabrication steps, see \cite{wiki:wafer}, where intermediate measurement of qualities  are not a feasible. In our concrete examples we treated up to $30$ different steps and more than $90$ machines. Unfortunately,  since our industrial partner  insists on the fulfillment of an NDA we are not allowed to publish any more information about  the project. Probably for the same reasons, it was impossible for us to find public available data on other industrial projects. 

\subsection{Simulations}
We consider the DAG network as given in Figure \ref{fig:DAG1} and consider  Gaussian qualities as described in Example \ref{ex:gaussian}. The matrix $S$ consists in this case of independent Gaussian random variables. In this simulation we consider the distribution that is characterized by
\begin{equation}
\E [S] := \begin{bmatrix}
  0 & 0 & 0 & 0 \\ 
  0 & 2 & 0 & 0 \\ 
  0 & 0 & * & 0 \\ 
  0 & * & * & 0 \\ 
  \end{bmatrix}\mbox{ and }
  \V[S] := \begin{bmatrix}
  1 & 1 & 4 & 1 \\ 
  1 & 1 & 1 & 1 \\ 
  1 & 1 & * & 1 \\ 
  1 & * & * & 1 \\ 
  \end{bmatrix},
\end{equation}
where the $*$ are placeholders for the machines that do not exist. We simulate $n=200$ observations and obtain the following results for the estimators $T$ and $\Sigma$ (rounded to two decimals):
\begin{equation} T=
\begin{bmatrix}
  0.60 & -0.03 & 0.75 & 0.70 \\ 
  0.25 & 2.06 & 0.40 & 0.89 \\ 
  0.74 & -0.10 &  * & 0.38 \\ 
  0.69 & *  & * & 0.14 \\ 
  \end{bmatrix}, \quad
\hfill
\Sigma= \begin{bmatrix}
  6.46 & 5.62 & 8.92 & 7.12 \\ 
  8.76 & 6.13 & 5.38 & 4.89 \\ 
  4.97 & 6.58 & *  & 8.41 \\ 
  7.37 &  * & * & 8.14 \\ 
  \end{bmatrix}.
\end{equation}
While the difference in mean in the second column seems to be obvious, the difference of the variances in the third  column might be overlooked and differences in the variance in the first and last columns could be suspected.

We perform pairwise $t$-tests for each column with Benjamini-Yekutieli adjustment.  While in all but the second column no statistically significant difference is detected, the difference in mean of the second machine in column $2$ is detected with a $p$-value of $9.8 \cdot 10^{-6}$.  
The Bartlett test does not find any differences in the variances in columns $1,2,$ and $4$. However, the difference of variance in column $3$ seems to be statistically significant with a $p$-value of $0.01$. The number of observations of $200$  is rather small compared to the number of different paths, which is $96$. This explains the fact that the matrix $\Sigma$ does not reflect the correct differences of variances; in particular we have in the third column that $8.92-5.38=3.54 > 3 = 4-1$. In order to demonstrate the convergence of the estimators we calculate the estimators for $n=10000$:
\begin{equation} T=
\begin{bmatrix}
  0.63 & 0.02 & 0.62 & 0.65 \\ 
  0.62 & 1.99 & 0.68 & 0.59 \\ 
  0.66 & -0.06 & * & 0.69 \\ 
  0.70 & * & * & 0.67 \\ 
  \end{bmatrix}, \quad
\hfill
\Sigma= \begin{bmatrix}
  6.32 & 5.47 & 7.82 & 6.44 \\ 
  6.28 & 5.42 & 4.85 & 6.03 \\ 
  6.26 & 5.39 &  * & 6.23 \\ 
  6.44 &  * & * & 6.59 \\ 
  \end{bmatrix}.
\end{equation}

\subsubsection*{Acknowledgment}
The authors wish to thank Alessandro  Chiancone, Herwig Friedl,  J\'er\^{o}me Depauw, and Marc Peign\'e for stimulating discussing during this project.   A.~G.~acknowledges  financial  support  from  the  Austrian  Science  Fund  project  FWF  P29355- N35.  Grateful acknowledgement is made for hospitality from TU-Graz  where the research was carried out during visits of S.~M. 

\bibliographystyle{abbrv}
\bibliography{bib_rca}

\begin{thebibliography}{10}

\bibitem{BaGu:09}
J.~Bang-Jensen and G.~Gutin.
\newblock {\em Digraphs}.
\newblock Springer Monographs in Mathematics. Springer-Verlag London, Ltd.,
  London, second edition, 2009.
\newblock Theory, algorithms and applications.

\bibitem{BeYe}
Y.~Benjamini and D.~Yekutieli.
\newblock The control of the false discovery rate in multiple testing under
  dependency.
\newblock {\em Ann. Statist.}, 29(4):1165--1188, 2001.

\bibitem{BoGeWeAr:10}
D.~{Bohme}, M.~{Geimer}, F.~{Wolf}, and L.~{Arnold}.
\newblock Identifying the root causes of wait states in large-scale parallel
  applications.
\newblock In {\em 2010 39th International Conference on Parallel Processing},
  pages 90--100, 2019.

\bibitem{wiki:wafer}
W.~contributors.
\newblock {W}afer fabrication, 2019.

\bibitem{DrPu:96}
N.~R. Draper and F.~Pukelsheim.
\newblock An overview of design of experiments.
\newblock {\em Statist. Papers}, 37(1):1--32, 1996.

\bibitem{FeRuSch:97}
D.~G. Feitelson, L.~Rudolph, U.~Schwiegelshohn, K.~C. Sevcik, and P.~Wong.
\newblock Theory and practice in parallel job scheduling.
\newblock {\em Lecture Notes in Computer Science}, pages 1--34, 1997.

\bibitem{Gu:09}
A.~Gut.
\newblock {\em Stopped random walks}.
\newblock Springer Series in Operations Research and Financial Engineering.
  Springer, New York, second edition, 2009.
\newblock Limit theorems and applications.

\bibitem{HaLe:89}
S.~T. Hackman and R.~C. Leachman.
\newblock A general framework for modeling production.
\newblock {\em Management Science}, 35(4):478--495, Apr 1989.

\bibitem{KlZe:08}
C.~Kleiber and A.~Zeileis.
\newblock {\em Applied Econometrics with R}.
\newblock Springer New York, 2008.

\bibitem{KoPh:01}
V.~S. Kouikoglou and Y.~A. Phillis.
\newblock {\em Hybrid Simulation Models of Production Networks}.
\newblock Springer US, 2001.

\bibitem{MoPeVi:01}
D.~C. Montgomery, E.~A. Peck, and G.~G. Vining.
\newblock {\em Introduction to linear regression analysis}.
\newblock Wiley Series in Probability and Statistics: Texts, References, and
  Pocketbooks Section. Wiley-Interscience, New York, third edition, 2001.

\bibitem{MoDoBe:14}
D.~Mourtzis, M.~Doukas, and D.~Bernidaki.
\newblock Simulation in manufacturing: Review and challenges.
\newblock {\em Procedia CIRP}, 25:213 -- 229, 2014.
\newblock 8th International Conference on Digital Enterprise Technology - DET
  2014 Disruptive Innovation in Manufacturing Engineering towards the 4th
  Industrial Revolution.

\bibitem{Pe:09}
J.~Pearl.
\newblock {\em Causality}.
\newblock Cambridge University Press, 2009.

\bibitem{Po:18}
S.~Popov.
\newblock The {T}angle, version April 30, 2018.

\bibitem{Schu:05}
M.~Schulz.
\newblock Extracting critical path graphs from mpi applications.
\newblock In {\em 2005 IEEE International Conference on Cluster Computing},
  pages 1--10, 2005.

\bibitem{Se:18}
D.~Selvamuthu and D.~Das.
\newblock {\em Introduction to Statistical Methods, Design of Experiments and
  Statistical Quality Control}.
\newblock Springer Singapore, 2018.

\bibitem{Sk:08}
S.~S. Skiena.
\newblock {\em The Algorithm Design Manual}.
\newblock Springer London, 2008.

\bibitem{Sz:15}
C.~{Szegedy}, {Wei Liu}, {Yangqing Jia}, P.~{Sermanet}, S.~{Reed},
  D.~{Anguelov}, D.~{Erhan}, V.~{Vanhoucke}, and A.~{Rabinovich}.
\newblock Going deeper with convolutions.
\newblock In {\em 2015 IEEE Conference on Computer Vision and Pattern
  Recognition (CVPR)}, pages 1--9, 2015.

\end{thebibliography}

\bigskip
\noindent\begin{minipage}{0.48\textwidth}
Abraham Gutierrez \newline
Institute of Discrete Mathematics,\newline
Graz University of Technology\newline
Steyrergasse 30,\newline
8010 Graz, Austria\newline
\texttt{a.gutierrez@math.tugraz.at}
\end{minipage}
\hfill
\begin{minipage}{0.48\textwidth}
Sebastian M\"{u}ller\newline
Aix Marseille Universit\'e\newline
CNRS,  Centrale Marseille\newline 
I2M\newline
UMR 7373\newline 13453 Marseille, France\newline
\texttt{sebastian.muller@univ-amu.fr}
\end{minipage}
\end{document}